\documentclass[12pt]{article}
\usepackage{amsmath}
\usepackage{amssymb}
\usepackage{epsfig}
\usepackage{textcomp}
\usepackage{amsthm}

\newtheorem{theorem}{Theorem}
\newtheorem{definition}[theorem]{Definition}
\newtheorem{lemma}[theorem]{Lemma}
\newtheorem{corollary}[theorem]{Corollary}

\usepackage{marginnote}
\usepackage[marginparwidth=2cm]{geometry}

\font\german=eufm10 at 10pt \def\Buchstabe#1{{\hbox{\german #1}}}
\def\EA{\Buchstabe{A}} 

\newcommand{\cC}{\mathcal{C}}
\newcommand{\tcC}{\widetilde{\mathcal{C}}}
\newcommand{\tcCb}{\widetilde{\mathcal{C}}_b}

\newcommand{\cH}{\mathcal{H}}
\newcommand{\cM}{\mathcal{M}}

\newcommand{\cS}{\mathcal{S}}
\newcommand{\cT}{\mathcal{T}}

\newcommand{\QAB}{Q^{1+AB}}

\newcommand{\la}{\langle}
\newcommand{\ra}{\rangle}
\newcommand{\tcM}{\widetilde{\mathcal{M}}}
\newcommand{\tcMb}{\widetilde{\mathcal{M}}_b}
\newcommand{\Om}{\Omega}

 at10pt

\begin{document}

\title {A histories perspective on characterising quantum non-locality}


 \author{Fay Dowker\footnote{Imperial College, Blackett Laboratory, Prince Consort Road, London, SW7 2BZ, United Kingdom.}, Joe Henson$^*$,and Petros Wallden\footnote{Heriot-Watt University, Edinburgh EH14 1AS, United Kingdom; University of Athens, Physics Department, Panepistimiopolis 157-71, Athens, Greece.} }

  \maketitle
\begin{abstract}

We introduce a framework for studying non-locality and contextuality
inspired by the path integral formulation of quantum theory. We prove that 
 the existence of a strongly positive joint quantum measure  -- the quantum analogue of 
 a joint probability measure -- on a set of experimental probabilities implies the Navascues-Pironio-Acin
 (NPA) condition $Q^1$ and is implied by the stronger NPA condition $Q^{1+AB}$.
 A related condition is shown to be  equivalent to $Q^{1+AB}$. 
  

\end{abstract}
\vskip 1cm

\section{Introduction}

The phenomenon now commonly referred to
as non-locality was, from the early days of quantum mechanics, 
central to debates on the adequacy of the standard formalism of quantum theory and other conceptual issues (see for example \cite{Bacciagaluppi:2006dn}).  More recently, there has been 
interest in probing the exact quantitative extent of non-locality in 
quantum theory.  For, while ordinary quantum models of experiments can violate Bell's local causality condition \cite{Bell:1990}, they do not allow all sets of experimental probabilities or ``behaviours'' \cite{Tsirelson:1980, Tsirelson:1987, Tsirelson:1993} consistent with the weaker, operational condition, ``no superluminal signalling''   \cite{Popescu:1994}. A research program of  ``characterising quantum non-locality'' has arisen with two closely related goals:  to provide a method of determining whether a given experimental behaviour could have been produced by an ordinary quantum model, and  to discover physical or information-theoretic principles that result in constraints on possible behaviours.

Building on the pioneering work of Tsirelson,  Navascues, Pironio and Acin (NPA) made substantial progress towards the first goal. NPA introduced a sequence of conditions of increasing strength, satisfaction of all of which implies the existence of an ordinary quantum model for a behaviour. Most 
importantly, each condition can be decided -- given unlimited computing power -- using standard programming techniques \cite{Navascues:2007,Navascues:2008}.    As to the second goal, a number of proposals have been made for information-theoretic principles that result in bounds on behaviours, including ``non-triviality of information complexity'' \cite{vanDam:2005}, ``no advantage for non-local computation'' \cite{Linden:2007}, ``information causality'' \cite{Pawlowski:2009}, ``macroscopic locality'' \cite{navascues:2009}, 
and ``local orthogonality''/``consistent exclusivity'' \cite{Foulis:1981, Fritz:2012, Cabello:2012}.  

Progress on these issues could provide a more general understanding of the limitations on measures of success for quantum information processing tasks, along with new ways of bounding them. It could also throw light on the conceptual issues mentioned above.  One view of the conflict between locality and quantum theory is that we should seek to uphold some essential feature of Bell's notion of local causality beyond mere no-signalling, while shedding some (at present unidentified) conceptual ``excess baggage.'' Drawing an analogy with the apparent conflict between relativity of motion and constancy of the speed of light which resolved itself in the theory of relativity, the hope is that a significant gain in understanding might result from this rehabilitation of causal ideas in quantum theory: see, for example,  \cite{Leifer:2011} for an approach to quantum theory influenced by such considerations.  Developing a understanding of what restrictions on non-locality remain in quantum theory could aid this endeavour.

In the meantime, related questions have arisen more or less independently in work on the problem of quantum gravity in a path integral framework.  A number of quantum gravity researchers have argued that, rather than trying to canonically quantise 
GR, the best prospects for success lie in formulating quantum theory in 
a fundamentally relativistic way using the path integral (see \textit{e.g.} \cite{Hawking:1978jz, Hartle:1989, Sorkin:1997gi}). To understand the reasoning behind this, recall that the 
Hilbert space formalism for quantum theory arises from 
canonical quantisation of the Hamiltonian form of classical mechanics: in order to perform a canonical quantisation of general relativity (GR) it must be formulated as a Hamiltonian theory based on a foliation of spacetime with spacelike hypersurfaces of constant time coordinate. 
General covariance implies that such a time coordinate and such hypersurfaces 
are not physical and the canonical quantisation of this Hamiltonian form of GR leads to the infamous ``problem of time'' \cite{Isham:1992, Kuchar:2011} or, more accurately, multiple \textit{problems} of time 
\cite{SorkinPirsa:2008}. 
These difficulties suggest that quantum gravity requires a formulation of quantum theory
based fundamentally on physical, covariant, spacetime quantities. Early in the history of QM,
Dirac \cite{Dirac:1933} argued that the canonical theory, based as it is on a choice of time variable, is 
essentially non-relativistic and that the alternative is to base quantum mechanics on the Lagrangian form of 
classical mechanics which ``can easily be expressed relativistically on account of the
action functional being a relativistic invariant.'' Dirac showed that this leads to the path integral or sum-over-histories which was brought to prominence by Feynman in his paper, ``A \textbf{spacetime} approach  to non-relativistic quantum mechanics'' (our emphasis) \cite{Feynman:1948}. 
An approach to quantum theory based on the path integral and free from any fundamental reference to states or measurements was pioneered by Caves, Hartle, Sorkin and others \cite{Caves:1986, Gell-Mann:1989nu, Hartle:1992as, Hartle:1991bb, Sorkin:1994dt}.   The many advantages for quantum gravity of a histories approach over a canonical approach have been set out by Sorkin \cite{Sorkin:1997gi}.  

In quantum gravity we must decide which principles from current theories to preserve and which to jettison in the search for the new theory and a principle of relativistic causality would 
seem to be an excellent candidate for keeping. 
However, the explicitly operational,  no signalling condition is of no use in quantum theories of closed systems such as the whole universe as often considered in quantum gravity. 
 A principle analogous to Bell's local causality, suitably generalised to apply in quantum theories of closed systems, which bans superluminal signalling as a
special case and yet allows violations of the Bell inequalities, is what is needed \cite{Henson:2005wb}. Most concretely, 
in causal set quantum gravity  Bell's local causality has been applied as a constraint on classical stochastic dynamical models of a growing discrete spacetime (causal set), with promising and suggestive results \cite{Rideout:2000a, Sorkin:1998hi}. If we can find a principle of \textit{quantum Bell causality},  quantum dynamical models of discrete causal set spacetimes 
could then be constructed using a similar approach.

 At this point, this important question in this approach to quantum gravity intersects with the 
program of characterising quantum non-locality in quantum information research
in the following way.
There is a trio of equivalent conditions on a behaviour in 
the Clauser-Horne-Shimony-Holt (CHSH) scenario: the existence of a Bell locally causal model  (also known as a ``classical 
screening off model'') of the probabilities, the satisfaction of all the CHSH (Bell) inequalities 
 \cite{Clauser:1969a} and the existence of a joint probability measure on all the possible (counterfactual) outcomes \cite{Fine:1982, Craig:2006ny}. This is also the case for more general experimental scenarios, with appropriate analogues, usually referred to simply as Bell inequalities, of the
CHSH inequalities \cite{Werner:2001}.  Thus, Bell's local causality, obedience of certain bounds on experimental probabilities and a non-contextuality condition are equivalent.\footnotemark 
The path integral approach to QM reveals quantum theories to be generalised measure theories  in which probability measures are generalised to  \textit{quantum measures} by weakening the Kolmogorov sum rule to an analogous rule that expresses the absence of genuine three-way interference between alternative histories for a system \cite{Sorkin:1994dt}. This immediately raises the question of what becomes of the ``Fine trio'' of equivalent conditions -- local causality (classical screening off), bounds on experimental  probabilities and non-contextuality -- when we enlarge the category of physical models we 
consider to ones in which dynamics is represented by quantum measures.  
A natural analogue of non-contextuality is the existence of a \textit {joint quantum measure} over all the  outcomes \cite{Craig:2006ny}. The
existence of such a measure in the CHSH scenario implies no restriction on the 
experimental probabilities,
indeed a Popescu-Rohrlich (PR) box is possible \cite{Craig:2006ny,Barnett:2007}. However, requiring the joint quantum measure to be \textit{strongly positive} implies the Tsirelson inequalities on the correlations \cite{Craig:2006ny,Barnett:2007} and
thus we have the beginnings of a quantum analogue of the Fine trio, with an
implication of the form  ``joint measure implies bounds on probabilities.'' This raises the hope that 
this will be a guide to discovering the quantum Bell causality condition, the
most physically important condition in the trio. In particular, quantum Bell causality should
be a condition that implies the existence of a  strongly positive joint quantum measure. 
One proposal for such a condition made in \cite{Craig:2006ny} has not, thus far, been 
successful in this regard.

 \footnotetext{This equivalence is the reason that any behaviour that violates a
 Bell inequality is commonly called ``non-local''  and why phrases such as ``characterising non-locality''
have come to be used for seeking and studying conditions that limit Bell inequality-violating behaviours. We have adopted this terminology in this paper in order to make contact with much of the literature but we note here that saying that a behaviour is non-local is a loose short hand for the precise meaning that the 
experimental probabilities cannot arise from a Bell locally causal model.}

In this article we continue the investigation of the quantum Fine trio. We have not 
achieved the goal of discovering quantum Bell causality but report on 
progress in understanding the quantitative strength of the condition of 
the existence of a strongly positive joint quantum measure.
We use the
the machinery of the NPA hierarchy and  show
that the strongly positive joint quantum measure condition closely resembles conditions from the 
hierarchy.  This allows the set of behaviours satisfying the 
quantum measure condition to be positioned rather precisely within the hierarchy.

\section{Joint measurement scenarios}
\label{s:framework}
 
We describe here the framework of \textit{joint measurement scenarios}.
It makes contact with the path integral/sum-over-histories approach to quantum mechanics, is suitable for discussing and resolving questions of common interest 
between quantum information, foundations of quantum mechanics and 
quantum gravity and is closely related to the formalism proposed by Liang, Spekkens and Wiseman  (LSW) \cite{Liang:2011}.The framework can be conceived of 
purely operationally but the measurement outcomes which which it deals
can also be thought of as physical events which are to be (eventually)
included in the physics -- quantum or otherwise. 

In this conceptual framework, there is a set of devices which can perform measurements.  To each
measurement is associated an exhaustive set of possible, exclusive outcomes.  Certain sets of measurements are compatible, that is they can be jointly performed.  It may be useful to have in mind, for instance, spin measurements on a collection of atoms or polarisation measurements on some photons.  Let $\cM := \{M_1, M_2, \dots M_p\}$ be the set of possible \textit{basic measurements}. Let  $\Xi_i$ be the set of labels for the possible outcomes for the measurement $M_i$.  We will assume $\Xi_i$ has finite cardinality for all $i= 1,\dots p$.
We introduce the \textit{non-contextuality space} (NC space),  $\Xi$, an element of which specifies an outcome for each of the measurements $M_i$ even though they cannot, in general, be jointly performed:
\begin{equation}
\label{e:xi_def}
\Xi:=\Xi_1 \times \Xi_2 \times ... \times \Xi_p\,.
\end{equation}

Each basic measurement, $M_i$, corresponds to a partition of $\Xi$.  Each possible outcome of the measurement $M_i$ will be represented by one of the subsets in this partition, $X \in M_i$. 
Note that we refer to the measurement 
and the partition with the same symbol. 
For a basic measurement $M_i$, an outcome $X_i\in M_i$   is a set of the form
\begin{equation}
\label{e:basic_outcome}
X_i  =\Xi_1 \times  \dots \Xi_{i-1}\times \{ a_i \} \times \Xi_{i+1} \dots \times \Xi_p \, \subset \, \Xi\,,
\end{equation}
where $a_i \in \Xi_i$.  The set of all such subsets, ranging over all 
elements of $\Xi_i$,  is the partition of $\Xi$ identified with basic measurement $M_i$.

A joint measurement is specified by a set of basic measurements $\{M_{i_1},M_{i_2},\dots M_{i_n}\} $
 and the partition, 
$M$, corresponding to the joint measurement, is the set of all sets of the form
\begin{equation}
\label{e:joint_outcome}
X_{i_1} \cap X_{i_2} \cap \dots X_{i_n}\,,
\end{equation}
where $X_{i_k} \in M_{i_k}$.  Note that each outcome of one of the basic measurements $M_{i_k}$ 
 is a union of outcomes for 
the  joint measurement $M$. Also, for each joint outcome $X\in M$, and each outcome $X_{i_k}$ of basic
measurement $M_{i_k}$,  either $X \cap X_{i_k} = X$ or $X \cap X_{i_k} = \emptyset$. 

A \textit{joint measurement scenario}, based on the set of basic measurements $\cM$ and NC space $\Xi$, 
is fixed by specifying a collection of joint measurements,  to be thought of as those joint measurements which are physically possible to do. We will denote the set of all joint measurements in the scenario by $\tcM$. 
Each measurement in $\tcM$ corresponds to a subset of $\cM$. All the 
basic measurements are elements of $\tcM$ and
if there is a joint measurement in $\tcM$ of a set $K$ of basic measurements
then there is a joint measurement in $\tcM$ of each subset of $K$. 
Endowing $\tcM$ with the corresponding
order by inclusion and adding the empty measurement to $\tcM$ makes $\tcM$ 
a meet semilattice.  A maximal joint measurement is a 
measurement that is maximal in this order and the basic measurements are minimal in the order if the empty measurement is
removed.  We say that measurement $M$ is included in measurement 
$M'$ if  $M$ precedes $M'$ in the order. 
We say that a collection of measurements in $\tcM$ can be \textit{jointly performed}  if there is a joint measurement in $\tcM$ which includes them all.

Note that this formalism already includes some kind of non-contextuality assumption: 
an outcome, $X$, of a basic measurement, $M_i$ is dealt with by the formalism 
as the \textit{same} event, no matter which 
joint measurement $M_i$ is included in. This would be physically most justifiable were 
the basic measurements all to have specified spacetime locations such that 
jointly performable basic measurements are in different spacetime locations. Then the
non-contextuality is justified by a form of locality, 
namely the assumption that the basic measurement outcomes are \textit{local events}
whose description is in terms of local variables associated to the 
spacetime region in which the outcome occurs.\footnotemark  If the  situation 
being modelled is supposed \textit{not} to be of this form then the 
non-contextuality assumption inherent in the formalism has to be justified in some other 
way.  We will return to these considerations later. 

\footnotetext{Note that this assumption does not require that these local 
variables exhaust all there is to say in physics -- there may also be 
nonlocal physical variables -- all it requires is that there \textit{are} local variables
and local events and that the outcomes are local events \cite{Henson:2013xx}.}

For each measurement $M\in \tcM$, there is a set of outcomes which are the elements of the 
partition $M$. We will refer to these as the fine-grained outcomes of 
$M$.   We consider also all \textit{coarse-grained outcomes} 
\textit{i.e.} all unions of the fine-grained outcomes of $M$. For each $M\in \tcM$ the set of all coarse-grained outcomes
together with the empty outcome form  the (Boolean) sigma algebra, $\EA_M$, generated by the fine-grained outcomes. We will use the term sigma algebra even though the case in hand is finite because one
might want to extend the formalism to infinite scenarios. 
The union of the algebras of outcomes over all measurements in a scenario will be denoted $\tcC := 
\bigcup_{M\in \tcM} \,\EA_M$
and the union of the algebras over all basic measurements will be denoted $\cC := 
\bigcup_{i = 1}^p \, \EA_{M_i}$.

Both $\tcC$ and $\cC$ are subsets of the power set, $2^\Xi$ of $\Xi$ but neither are sigma algebras: in 
general they are not closed under union or intersection. While  $\tcC$ is a poset, ordered by inclusion, it is not a lattice.  However,  $\tcC$ is an orthocomplemented poset, meaning for each outcome $X \in \tcC$ there exists $\bar{X} \in \tcC$ such that: (i) $\bar{\bar{X}}=X$; (ii) $X \subset Y$ implies $\bar{Y} \subset \bar{X}$; (iii) $X \subset Y$ and $\bar{X} \subset Y$ implies $Y=\Xi$.

The NC space $\Xi$ together with the set $\tcM$ of all (doable) measurements specifies the joint measurement scenario $\cS=\{\Xi, \tcM\}$. This framework is neither the most restrictive nor the most general that one
might want to consider\footnotemark.  For example, one can imagine putting further constraints on the set of measurements, $\tcM$,
such as, if all pairs of measurements in a set of measurements $K\subset \cM$ are jointly 
performable, then  $K$ is jointly performable \cite{Randall:1974, Fritz:2012, Cabello:2012}. One could consider more general NC spaces and/or more general partitions of the NC space as measurements. 
For example, one could consider scenarios in which some measurement choices can depend on the 
outcomes of other measurements, something 
we will investigate in section \ref{s:branching_measurements}.

\footnotetext{In the Fritz-Leverrier-Sainz (FLS) framework \cite{Fritz:2013}, as in some earlier treatments \cite{Foulis:1981,Cabello:2010}, the basic structure is a set of outcomes along with a set of subsets of the outcomes representing measurements.  While the NC space could be defined as the set of ``deterministic probabilistic models'' (as defined by FLS) for an FLS scenario, measurements and their outcomes cannot always be expressed in terms of this set as they are above (it can be empty for valid FLS scenarios). There are more other differences in approach, most importantly that the FLS framework is not defined in terms of basic and joint measurements.  Instead, the no-signalling and the commutation requirements for quantum models are enforced in non-locality scenarios by adding extra measurements.  While this gives a simple and powerful formalism, without adding more structure it is difficult to define \textit{e.g.} the original NPA set $Q^1$, or to discuss different versions of our quantum measure-based conditions.  \label{f:fls}}

An \textit{experimental behaviour} is a joint measurement scenario $(\Xi, \tcM)$ together with a probability measure 
$P(\cdot | M)$ on the sigma algebra $\EA_M$, for each $M\in \tcM$. 
The probabilities in an experimental behaviour are required to be consistent \textit{i.e}, if an outcome $X \in \tcC$
is an element of $\EA_M$ and of $\EA_{M'}$ then 
$P(X|M) = P(X|M')$. 
A behaviour is therefore equivalent to a consistent \textit{probability function}  $P(\cdot)$ whose domain is the set of outcomes $\tcC$.
 $P$ is not a probability \textit{measure} because the set of outcomes $\tcC$ is not a sigma algebra, however 
 it is a measure when restricted to any algebra of outcomes $\EA_M$ for a measurement $M$.
$P(\cdot)$ can be referred to as the (set of) \textit{experimental probabilities}. 

The consistency of the 
probability function is a further non-contextuality  assumption beyond 
the one mentioned above. There it was stressed that the formalism treats the outcome, $X$, of a measurement, $M$, as the same event no matter what 
measurement, $M'$, 
$M$ may be included in. Now we are assuming that the probability of $X$ is the same 
no matter what joint measurement  $M$ is included in.  This would be most 
strongly, physically justified if the basic measurements have specific spacetime locations
and those locations are such that any jointly performable basic measurements are 
spacelike separated from each other. Then the 
consistency of the probability function would be justified by 
relativistic causality: it is a no-signalling 
condition.  Again, if the situation being modelled does not have this 
underlying spacetime structure,
then the assumption of the consistency of the probability function 
would have to be argued for in some other way.

  
 A key question in the analysis of contextuality and non-locality takes the form, in this 
 framework, of whether a given behaviour $(\Xi, \tcM, P)$ can be consistently extended in some way to
 a larger subset of $2^\Xi$ than $\tcC$. 
 
 \paragraph{Example: Bell scenarios} 
 To make contact with the familiar special case of non-locality setups we will need some additional concepts.

\textit{Composition} of a set of scenarios $\cT=$$\{\cS_1,\cS_1,...,\cS_n\}$ into one $\cS(\cT)=$$\{\Xi_\cT, \tcM_\cT \}$ can be defined if we are modelling a situation in which the choice of measurements on one system does not affect what can be done with the others: a particularly natural assumption if the systems are spacelike separated for the period of measurement.  The NC space in this case is 
\begin{equation}
\Xi_\cT = \Xi_{\cS_1} \times \Xi_{\cS_2} \times ... \times \Xi_{\cS_n}
\end{equation}
where  $\Xi_{\cS_i}$ is the NC space for scenario $\cS_i$.  For each outcome $X$ of a measurement $M$ in a scenario $\cS_i \in \cT$ there is a set $X_\cT \in \Xi_\cT$ defined as 
\begin{equation}
X_\cT =  \Xi_{\cS_1} \times ... \times \Xi_{\cS_{i-1}} \times X \times \Xi_{\cS_{i+1}} \times ... \times \Xi_{\cS_n},
\end{equation}
and the set
\begin{equation}
M_{\cT}=\{X_\cT\}_{X\in M}
\end{equation}
is a new measurement partition $M_\cT$ on $\Xi_\cT$. The set of measurements $\tcM_\cT$ contains all of these measurements for all members of $\cT$.  Because the operations on one system do not affect the possibilities of measurement elsewhere, any set of these measurements all of which come from different scenarios in $\cT$ can be jointly performed, and there are no further measurements in $\tcM_\cT$. 

\textit{Bell scenarios} are a special type of joint measurement scenario suitable for the study of non-locality.  First, consider a simple measurement scenario with $m$ mutually incompatible measurements, each with $d$ outcomes.  Composing $n$ of these ``wings'' as above gives a Bell scenario, specified by the triple $(n,m,d)$. Each ``local'' outcome for each wing $X_{a_i x_i}$ is thus uniquely specified by two integer indices, $x_i$ for the measurement ``setting'' choice and $a_i$ for the outcome, where $i$ labels the wing.  Standard notation refers to such outcomes simply as $(a_i|x_i)$.  Similarly, ``global'' outcomes of a joint measurement on all systems are refereed to as $(a_1,a_2,...,a_n\, |\, x_1,x_2,...,x_n)$.  The CHSH scenario, specified by $(2,2,2)$, is the simplest interesting Bell scenario.  It is the most well-studied example, and also the one to which some previous results exclusively apply.

For Bell scenarios, consistency of the probabilities of outcomes over different measurements is equivalent to no signalling: the probabilities for outcomes of a measurement will be independent of any measurement choices made at spacelike separation.


\section{Conditions on behaviours}

\subsection{The \textit{JPM} condition, non-contextuality and Bell's Local Causality}
\label{s:jqm}
A behaviour gives probabilities for the outcomes of doable experiments. Let us 
now consider a \textit{joint probability measure} $P_J$ on the sigma algebra, $\EA$, generated by $\tcC$ and
 define the set of behaviours consistent with the existence of such a measure. 
 In our case,  $\EA = 2^\Xi$ since $\Xi$ is finite.

\begin{definition} (JPM)
\label{d:jpm}
A behaviour $(\Xi, \tcM, P)$  is in the set \textit{JPM} if there exists a joint probability measure $P_J$ on the sigma 
algebra $2^\Xi$ such that\begin{equation}
P(X) = P_J(X), \quad \forall X \in \tcC\,. \label{e:joint_measure} 
\end{equation}
\end{definition}

Since $2^\Xi$ is finite, $P_J$ is fixed by its values for the atoms of the algebra $2^\Xi$, \textit{i.e.} the singleton sets corresponding to the elements of 
$\Xi$ itself: $P_J(\{\gamma\})$, $\gamma\in \Xi$.  So, if a behaviour is in JPM then 
$P(X) = \sum_{\gamma \in X} P_J(\{\gamma\})$. \textit{JPM} is therefore a non-contextuality condition. Each measurement outcome could be said to be determined by some pre-existing physical property which is not affected by, but merely revealed by the particular measurement performed.  The measure $P_J$ can be interpreted as a measure over these properties.  In the language of stochastic processes, outcomes are just some of the \textit{events}, subsets of the ``sample space'' $\Xi$, \textit{all} of which are now covered by one probability measure.

For Bell scenarios, the \textit{JPM} condition is equivalent to the existence of a Bell locally causal model \cite{Bell:1990} of the
probabilities. For the CHSH setup this was proved by Fine \cite{Fine:1982} (see \cite{Werner:2001} for a generalisation to all Bell scenarios, and further references). The JPM condition for Bell scenarios is also equivalent to the satisfaction of a set of (scenario dependent) linear inequalities between experimental probabilities, generically known as ``Bell inequalities'' \cite{Werner:2001}.   The equivalence between non-contextuality, Bell's local causality and constraints on experimental behaviours inspires our proposals at the quantum level, where we seek analogues of these conditions.

\subsection{Ordinary quantum behaviours}

The set of behaviours \textit{Q} that can be derived from ordinary quantum models is of obvious interest.  In the following, it will be useful to recall that the expression $X\in M$ identifies $X$ as
a fine-grained outcome of the measurement $M$ ($X$ is a subset of $\Xi$ in the measurement partition $M$) while $X\in \EA_M$ refers to any outcome in the Boolean sigma algebra generated by the fine-grained outcomes. 
  
\begin{definition} (Q)
\label{d:q}
A behaviour $(\Xi, \tcM, P)$ is in the set \textit{Q} if there exists an ordinary quantum model for the behaviour, that is, a Hilbert space $\cH$, a pure state $|\psi \ra \in \cH$, and a projection operator $E^X$ on $\cH$ for each fine-grained outcome, $X$, of each basic measurement such that
\begin{itemize}
\item For each basic measurement $M_i\in\cM$,
\begin{equation}
\label{e:qdef_projsum}
\sum_{X \in {M_i}} E^{X}=\mathbf{1}
\end{equation}
where $\mathbf{1}$ is the identity on $\cH$;
\item If basic measurements $M_i,M_j\in \cM$ can be jointly performed, then
\begin{equation}
\label{e:q_commute}
[E^{X},E^{Y}]=0 \quad \forall \, X \in M_i, Y \in M_j\,;
\end{equation}
\item
 If outcome $X$ of a joint measurement of the set of
basic measurements $\{M_{i_1},M_{i_2},...M_{i_n}\}$ is $X = \cap_{k = 1}^n X_{i_k}$
where $ X_{i_k} \in M_{i_k}$ then 
\begin{equation}
\label{e:oqb_probs} 
\la \psi | E^{X_{i_1}}E^{X_{i_2}}...E^{X_{i_n}} |\psi \ra = P(X)\,.
\end{equation}
\end{itemize}
\end{definition}

It follows that, for all basic measurements $M_i \in \cM$, the projectors defining an ordinary quantum model
satisfy:
\begin{equation}
\label{e:q_orthogonality}
E^X E^Y = \delta_{X Y} E^X \quad \forall \, X,Y\in M_i\, ,
\end{equation}
Since, from (\ref{e:qdef_projsum}) we have that $\sum_{Y \in M} E^X E^Y E^X = E^X$ for all $X \in M$, from which it follows that $E^X E^Y E^X =0$ for all $X \neq Y$, because all these terms are positive operators.

\begin{lemma}
\label{l:q2}
A behaviour $(\Xi, \tcM, P)$ is in the set \textit{Q} if and only if 
there exists a Hilbert space $\cH$, a pure state $|\psi \ra \in \cH$, and a map 
$E: \tcC \rightarrow L$, from the set of outcomes to the lattice, $L$, of projection operators 
on $\cH$ (equivalently, the lattice of subspaces of $\cH$) such that 
\begin{itemize}
\item
for each measurement $M\in\tcM$,  the image of $E$ restricted to $\EA_M$, 
$E(\EA_M)$, is a Boolean algebra and the restriction of $E$,
\begin{equation*}
E: \, \EA_M \rightarrow E(\EA_M)\,,
\end{equation*}
is a Boolean algebra homomorphism\,;
\item the probabilities are given by
\begin{equation}
\label{e:oqb_lattice_probs} 
 \la \psi | E(X) |\psi \ra = P(X), \ \ \forall X \in \tcC\,.
\end{equation}
\end{itemize}
\end{lemma}

\begin{proof} If the behaviour is in \textit{Q} then we have projection operators for each of the 
fine-grained outcomes, $X_{i_k}$, of each basic measurement
so define $E(X_{i_k}) : = E^{X_{i_k}}$. We then define a projector for a 
fine-grained outcome of a joint measurement by forming  the product of the 
projectors for the relevant basic measurement outcomes: 
for outcome $X$ of a joint measurement of the set of
basic measurements $\{M_{i_1},M_{i_2},...M_{i_n}\}$ is $X = \cap_{k = 1}^n X_{i_k}$
where $ X_{i_k} \in M_{i_k}$,
\begin{equation}
\label{e:joint_projector}
E(X):= E^{X_{i_1}}E^{X_{i_2}}...E^{X_{i_n}}\,,
\end{equation}
where the $E^{X_{i_k}}$ all commute and so this is indeed a projector.
Finally, 
the projector, $E(Y)$ for a coarse-grained outcome, $Y\in \EA_M$, is defined to be the sum of the projectors for the 
fine-grained outcomes of $M$ of which $Y$ is the union. This defines a map $E: \tcC \rightarrow L$
which has the properties above, by construction. 

If the map $E$ exists with the properties above, the projectors $E(X)$  for the 
fine-grained  outcomes of the basic measurements will satisfy the conditions in 
definition \ref{d:q}. Note here that we are using the definition of Boolean algebra homomorphism
which includes the condition that the unit element is mapped to the unit element. 
\end{proof} 

We will refer to such a Hilbert space and map $E$ as an \textit{ordinary quantum model} for the behaviour. 
Note that, throughout this paper, we 
assiduously refer  to such quantum models as ``ordinary''  in order to distinguish them from other models that, from a histories 
perspective, are also classified as ``quantum'', in particular  the quantum measure theory models introduced below. 

In Bell scenarios, the condition in definition \ref{d:q}, equation (\ref{e:q_commute}), translates to the requirement that the projectors for spacelike
separated measurements commute. The definition of ordinary quantum behaviours here follows the definition given in \cite{Navascues:2008} 
in which it is not required that the Hilbert space be a tensor product. See that reference for a discussion of the fact that it is not known in general whether requiring a tensor product Hilbert space is a strictly stronger condition (``Tsirelson's problem'').

A number of questions arise regarding \textit{Q}, including: what assumptions are sufficient to determine the shape of \textit{Q}; can we derive the condition based on simple, more-or-less physical criteria rather than by invoking the above definition, which is considerably less conceptually simple than the definition of \textit{JPM}; does the set of allowed behaviours remain the same when an alternative framework for quantum mechanics is considered, such as the path integral?

\subsection{NPA conditions}

Tsirelson pioneered the study of bounds on  \textit{Q}, which
can be thought of as defining sets of behaviours containing \textit{Q} \cite{Tsirelson:1980, Tsirelson:1987, Tsirelson:1993}. 
More recently, Navascues, Pironio and Acin  introduced a hierarchy of conditions 
resulting in a nested sequence of sets of behaviours for bipartite Bell scenarios which contain and converge to \textit{Q} \cite{Navascues:2007,Navascues:2008}.  The conditions defining these sets can be cast as semidefinite programs which facilitates solution by computer, and allows \textit{e.g.}~estimation of violations of Bell inequalities in QM, beyond the CHSH scenario.  Rather than review the whole hierarchy in detail here, two conditions that will be of particular relevance will be defined.  The definitions have been generalised to joint measurement scenarios but are equivalent to the original ones 
given by NPA when we restrict to bipartite Bell scenarios.

The set of behaviours $Q^1$ is the first in the original NPA hierarchy. 
 The basic observation used to define this set is that, in ordinary quantum models, the inner products between 
all pairs of vectors from the set $\{ |\psi\ra\}\cup\{E^X|\psi\ra\,; \, \mbox{ $X$ an outcome of a basic measurement}\}$  must obey certain conditions, including relations to the experimental probabilities. 
 If a model is in \textit{Q}, there must therefore exist a Hilbert space and vectors in it satisfying these conditions.  The set of behaviours obeying this condition is called $Q^1$.  This set of behaviours has received substantial attention in the literature.  In the CHSH scenario it is equivalent to a simple bound, and the maximum violation of several Bell inequalities for $Q^1$ behaviours are known \cite{Navascues:2008}.  It is also equivalent to the ``macroscopic locality'' condition proposed by Navascues and Wunderlich \cite{navascues:2009}. 
 \begin{definition} ($Q^1$) 
\label{d:q1}
A behaviour $(\Xi, \tcM, P)$ is in $Q^1$ if there exists a Hilbert space $\cH$ and vectors $\{ |X \ra \}_{X \in \cC}$ 
indexed by the coarse-grained outcomes of the \textbf{basic} measurements, that span $\cH$, and satisfy the following conditions:
\begin{itemize}
\item[(i)]  if $M_i$ is a basic measurement and $X,Y \in \EA_{M_i}$ are disjoint outcomes, $X\cap Y = \emptyset$, then
\begin{equation}
\label{e:q1sum}
| X \cup Y \ra= |X\ra + |Y \ra\,;
\end{equation}
\item[(ii)]  if basic measurements $M_i,M_j \in \cM$  are jointly performable and $X \in \EA_{M_i}$ and $Y \in \EA_{M_j}$, then
\begin{equation}
\label{e:q1joint_measurement}
\la X | Y \ra= P(X \cap Y)\;.
\end{equation}
\end{itemize}
\end{definition}

It is clear from the definition that $Q \subseteq Q^1$ and it is known that $Q^1$ is actually 
strictly larger than \textit{Q}  \cite{Navascues:2008}.

The definition of $Q^1$ is adapted directly from the NPA definition in terms of positive matrices, and is equivalent to it for bipartite Bell scenarios.
The proof is in Appendix \ref{a:npa_defs}.  For the discussion below, it will be easier to work directly with this definition expressed in terms of vectors. Note that the the condition that the vectors span $\cH$ can be dropped without altering the set $Q^1$ 
since if such a set of vectors exists that doesn't span the Hilbert space, then $\cH$  can be redefined to 
be the span of the vectors.

The basic measurements play a key role in the definition of $Q^1$ and an obvious  strengthening of the 
condition is to require that vectors corresponding to all outcomes of all joint measurements exist satisfying the appropriate conditions. This leads to 
the set of  behaviours called $\QAB$ in the NPA approach,\footnotemark  
 and is the first in the modified version of the hierarchy introduced later by Fritz, Leverrier and Sainz \cite{Fritz:2013}.  

\footnotetext{In the special case of the (2,m,d) Bell scenarios, the only joint measurements are joint measurements of one measurement on A's wing and one on B's wing of the experiment, hence the $AB$ in $Q^{1+AB}$.}

\begin{definition} ($\QAB$)
\label{d:tq1}
A behaviour $(\Xi, \tcM, P)$ is in $\QAB$ if there exists a Hilbert space $\cH$ and a set of vectors $\{ |X \ra \}_{X \in \tcC}$ 
indexed by the outcomes of the joint measurements, that span $\cH$  and satisfy the following conditions:
\begin{itemize}
\item[(i)]  if $M$ is a measurement and $X,Y \in \EA_{M}$ are disjoint outcomes, $X\cap Y = \emptyset$, then
\begin{equation}
\label{e:tq1sum}
| X \cup Y \ra= |X \ra + |Y \ra\,;
\end{equation}
\item[(ii)]  for each measurement $M\in \tcM$ and all  outcomes $X, Y \in  \EA_M$ 
\begin{equation}
\label{e:tq1joint_measurement}
\la X | Y \ra= P(X \cap Y)\; ;
\end{equation}
\item[(iii)] if  $X, Y \in \tcC$ and there exists a basic measurement $M_i$ and outcomes 
$X', Y' \in M_i$ such that $X' \cap Y' = \emptyset$, $X \subset X'$ and  $Y \subset Y'$ then 
\begin{equation}\label{e:tq1orthogonality}
\la X | Y\ra = 0\,.
\end{equation}
\end{itemize}
\end{definition}

Crucially,  condition (iii) says that two vectors corresponding to outcomes $X$ and $Y$ in $\tcC$ are orthogonal if they  imply disjoint outcomes for any basic measurement, 
\textit{even if the two outcomes $X$ and $Y$ are not compatible and are not both elements of any 
single measurement sigma algebra $\EA_M$}. 

The set $\QAB$ contains \textit{Q}, and it follows easily from the definitions that $Q^1$ contains $\QAB$.  There is strong computational evidence that these are strict containments \cite{Navascues:2008}.
The condition $\QAB$ is well-studied and has a number of remarkable properties. For instance, it can be shown to be equivalent to a bound on the Lov\'asz number of the ``orthogonality graph'' associated to the behaviour \cite{Cabello:2010, Fritz:2013}, an important quantity for classical information theory that was defined well before the advent of quantum information. Also, all of the currently proposed information-theoretic principles that restrict non-locality can be shown to be implied by the $\QAB$ bound \cite{Acin:2013}. 

Although the higher levels of the NPA hierarchy will not enter into the following discussion, note that level $n$ is equivalent to the existence of vectors with the same inner product properties as those obtained from ordinary quantum models by applying all \textit{products} of $n$ projectors to the state $|\psi\ra$ (the projectors corresponding to basic measurements for the NPA version, and joint measurements for the FLS version of the hierarchy).

\section{Conditions on behaviours from the quantum measure}

\subsection{Quantum measure theory}
\label{s:histories}

We briefly review the histories approach to quantum theory. For more details see \cite{Sorkin:1994dt,
Sorkin:1995nj, Dowker:2010ng, SorkinPathIntegral}. In this approach
quantum 
theory is understood as a species of generalised measure theory, a generalisation of classical stochastic processes
such as Brownian motion.  In histories-based theories -- classical, quantum or transquantum --
the fundamental kinematical concept is the set of
spacetime histories, $\Om$. For a particular physical system, 
each history in $\Om$ is as complete a
description of the system as is conceivable in the theory, over all time of interest.  
The kind of elements in $\Omega$ varies from theory to theory:
in $n$-particle classical or quantum mechanics, a history is a set of $n$
trajectories; in a classical or quantum scalar field theory, a history is a real or complex
function on spacetime; in GR, a history is a Lorentzian geometry. 
Discovering the appropriate set of histories $\Om$ for a particular system is part of the business of physics. 
In quantum theory, conceived of as  \textit{quantum measure theory}, 
 $\Omega$ is the set over which the integration of the path integral takes place.
Even for the non-relativistic particle, the path 
integral remains to be defined rigorously as a genuine integral over paths \cite{Sorkin:2011sp}, something that 
will not concern us here as we restrict ourselves to finite systems.

Once the set of histories has been settled upon, any proposition about
the system is represented by a subset of $\Omega$. For example
in the case of the non-relativistic particle, if $R$ is a particular spacetime region, the proposition 
``the particle traverses $R$''
corresponds to the set of all spacetime trajectories 
which intersect $R$.  We follow the standard terminology of stochastic
processes and refer to such subsets of $\Omega$ as \emph{events}.

In this framework, a physical theory based on $\Omega$ 
is expressed as a generalised measure theory, 
specified by a sigma algebra, $\EA \subseteq 2^\Omega$, of  events and a
measure,  $\mu:\EA \rightarrow \mathbb{R}$, where $\mu$ is a non-negative 
function which encodes (a combination of) the dynamics and initial condition.
Sorkin has identified a hierarchy of measure theories defined by a 
sequence of strictly weakening conditions on the measure, $\mu$, of which the first 
is the Kolmogorov sum rule for probability measures. 
Classical stochastic theories such as Brownian motion therefore lie in the first level of the Sorkin 
hierarchy and classical deterministic theories are 
special cases. 

The second level in the Sorkin hierarchy is quantum measure theory. 
A \emph{quantum measure} on an event algebra $\EA$ is a map $\mu: \EA
\to \mathbb{R}$ such that:
\begin{enumerate}
\item for all $\alpha \in \EA$,  $\mu(\alpha) \geq 0$
(\emph{Positivity});
\item for all mutually disjoint $\alpha, \beta, \gamma \in \EA$, \begin{equation}  \label{qsr}
 \mu(\alpha \cup \beta \cup \gamma) - \mu(\alpha \cup \beta) -
\mu(\beta \cup \gamma) - \mu(\gamma \cup \alpha) +\mu(\alpha) +
\mu(\beta) + \mu(\gamma) = 0
\end{equation}
({\emph{Quantum Sum Rule}});
\item $\mu(\Omega)=1$ (\emph{Normalisation}).
\end{enumerate}
The quantum sum rule expresses the lack of interference between triples of histories, just as
the classical Kolmogorov rule expresses the lack of interference between pairs of histories. The 
higher levels of the Sorkin hierarchy generalise this to lack of interference between $k$ histories. 
Each level of the hierarchy contains the previous levels, for example, a classical  stochastic level 1 theory 
satisfies the level 2 condition. 

There are, broadly, two interpretational frameworks  for quantum measure theory  in the literature. 
In both frameworks, the formalism of histories, 
events and the decoherence functional is the complete physics of the system: histories-based quantum 
theory is quantum theory of a closed system, a quantum theory without external observers, agents or measuring devices. 
Observers, agents and measuring devices, should they exist, are described within the theory. 
In  \textit{decoherent or consistent histories} \cite{Griffiths:1984rx,Omnes:1988ek,Gell-Mann:1989nu} probabilities are considered to be 
fundamental and attention is focussed on subalgebras of $\EA$ on which the quantum measure restricts to 
a probability measure. The fact that there are (infinitely) many such subalgebras must then be 
grappled with \cite{Dowker:1994dd,PhysRevLett.75.3038}. In the \textit{co-event} framework \cite{Sorkin:2006wq,Sorkin:2007}, attention is focussed on events of measure  zero and an interpretation sought based on the maxim ``events of measure zero do not happen.''

In all known quantum theories, the quantum measure is not given directly but 
via the agency of a  \emph{decoherence functional} 
on the event algebra. 
A decoherence functional is a map $D: \EA \times \EA \to \mathbb{C}$ such that \cite{Hartle:1989, Hartle:1992as}
\begin{enumerate}
\item For all $\alpha, \beta \in \EA$, we have $D(\alpha,\beta) =
D(\beta,\alpha)^*$ (\emph{Hermiticity})\,.
\item For all $\alpha, \beta, \gamma \in \EA$ with $\beta \cap \gamma
= \emptyset$, we have $D(\alpha, \beta \cup \gamma) = D(\alpha,
\beta) + D(\alpha,\gamma)$ (\emph{Bi-additivity})\,.
\item $D(\Omega, \Omega)=1$ (\emph{Normalisation})\,.
\item For all $\alpha \in \EA$,  $D(\alpha, \alpha) \ge 0 $ (Positivity)\,.
\end{enumerate}

The existence of a quantum measure is 
equivalent to the existence of a decoherence functional. 
If $D: \EA \times \EA \to \mathbb{C}$ is a decoherence functional
then the map $\mu : \EA \to \mathbb{R}$ defined by
$\mu(\alpha):=D(\alpha,\alpha)$ satisfies the 
conditions of being a quantum measure. And conversely, if $\mu$ is a 
quantum measure, there exists a (non-unique) decoherence functional 
$D$ such that $\mu(\alpha) = D(\alpha, \alpha)$ \cite{Sorkin:1994dt}.  Given this,
 we will  refer to a  decoherence functional also as a quantum measure. 
 We call a triple $(\Omega, \EA, D)$ a quantum measure system. 

It is also the case that in all quantum theories describing known physics the decoherence 
functional satisfies the condition of strong positivity \cite{Martin:2004xi}:
a decoherence functional is \textit{strongly positive} (SP) if
for any finite collection of events, $\alpha_1,\dots \alpha_N$, $\alpha_i \in \EA$,
 the $N \times N$ matrix $D(\alpha_i,\alpha_j)$ is
positive semidefinite. For the special case of finite event algebras (which is the 
case of relevance in this paper) a decoherence functional is 
strongly positive if the matrix of entries of $D$ indexed by the set of atoms of the 
algebra is a positive matrix. 

\begin{lemma}
\label{l:compose}
The product of two independent finite, strongly positive quantum measure systems is a 
strongly positive quantum measure system.
\end{lemma}
\begin{proof}
The product decoherence functional on the atoms of the tensor product algebra is given by
the Hadamard product of the two decoherence functionals of the  subsystems. The 
Hadamard product of two positive matrices is positive.
\end{proof}

Strong positivity thus implies composability of subsystems and is related to complete 
positivity of quantum operations in the context of evolving systems in quantum measure theory
\cite{Martin:2004xi}. 
There is therefore a good reason to adopt strong positivity as a
physical condition on quantum measures, and if one 
were convinced by this reason, one might want 
to redefine quantum measures to include this condition. We will not do so in this 
paper as we want to discuss the consequences of adopting or not adopting the 
condition. 



\subsection{A quantum analogue of non-contextuality}

One response to the Bell Theorem, is to retreat to talk of the experimental settings and outcomes only, and to 
investigate operational principles that constrain the experimental probabilities in order to 
try to understand the bounds
on \textit{Q}.  Loosely speaking, this is to start from the most general set of all behaviours and ask why the quantum constraints are so strong \cite{Popescu:1994}.  We propose a complementary perspective, continuing to 
focus on the NC space $\Xi$ for a behaviour, but from a histories perspective within the 
framework of quantum measure theory.  The approach then becomes one of generalising the 
non-contextuality condition \textit{JPM}, asking why the quantum constraints are so weak.

The \textit{JPM} condition for a behaviour is that there exists a joint probability measure on the full sigma algebra $\EA$ generated by $\tcC$. Probability measures are Sorkin level one measures and since 
quantum measures are Sorkin level two measures this suggests the following
obvious quantum analogue of JPM: 
\begin{definition} (JQM)
\label{d:jqm}
A behaviour $(\Xi, \tcM, P)$ is in the set \textit{JQM} if there exists a joint quantum measure, $\mu_J$ on sigma 
algebra $2^\Xi$ such that 
\begin{equation}
\mu_J(X) =  P(X), \quad \forall \ X\in\tcC \quad \textnormal{(Experimental probabilities)}\,.
\end{equation}
\end{definition}
This condition says that the marginal quantum measure --  $\mu_{J}$ restricted to
 the subalgebra $\EA_M$ --  for each measurement 
$M$, coincides with the behaviour's probability measure $P(\cdot| M)$ on $\EA_M$.
\textit{JQM} is a quantum non-contextuality condition.

We can re-express the JQM condition in terms of decoherence functionals:
\begin{lemma}
\label{d:jqdf}
A behaviour $(\Xi, \tcM, P)$ is in the set \textit{JQM} if and only if there exists a joint decoherence functional, $D_J$ over $2^\Xi$ such that 
\begin{itemize}
\item $D_J(X, X) =  P(X), \quad \forall \ X\in\tcC$ (Experimental probabilities);
\item$ D_J(X, Y) = 0, \quad \forall \ X, Y \in \EA_M, \ X\cap Y = \emptyset, \quad \forall \ M\in \tcM$ (Decoherence of 
alternative experimental outcomes).
\end{itemize}
\end{lemma}
\begin{proof} Suppose a behaviour is in \textit{JQM} so there exists a 
quantum measure $\mu_J: 2^\Xi \rightarrow \mathbb{R}$ such that 
$\mu_J(X) = P(X)$ for all $X\in \tcC$. Then, there exists a real decoherence functional $D_J$
such that $D_J(X, X) =  \mu_J(X)$ for all  $X\in 2^\Xi$  \cite{Sorkin:1994dt}. Any decoherence functional satisfies
\begin{equation}
2 \,\textnormal{Re}\left[ D(X, Y) \right]=  D(X \cup Y, X \cup Y) - D(X, X) - D(Y, Y) \,,
\end{equation} 
for any two disjoint events, $X$ and $Y$.
For any measurement, $M$ and  $X, Y \in \EA_M, \ X\ne Y$ we also have
\begin{equation}
D_J(X \cup Y, X \cup Y) =  P(X\cup Y), \quad  D_J(X,X) = P(X) 
\quad \textnormal{and}\quad  D_J(Y,Y) = P(Y)\,.
\end{equation}
Since $D$ is real, the decoherence condition 
$D_J(X,Y) = 0$ follows.  Note that the two conditions together are equivalent to 
$ D_J(X, Y) = P(X\cap Y) , \quad \forall \ X, Y \in \EA_M,  \quad \forall \ M\in \tcM$.

The converse is immediate on defining $\mu_J(X) = D_J(X,X)$
for all $X\in 2^{\Xi}$.  
\end{proof}

An obvious strengthening of the \textit{JQM} condition is:
\begin{definition} (SPJQM)
\label{d:spjqm}
A behaviour $(\Xi, \tcM, P)$ is in the set \textit{SPJQM} if there exists a joint decoherence functional, $D_J$ over $2^\Xi$ such that 
\begin{itemize}
\item $D_J(X, Y) = P(X\cap Y) , \quad \forall \ X, Y \in \EA_M,  \quad \forall \ M\in \tcM$;
\item $D_J$ is strongly positive (SP).
\end{itemize}
\end{definition}

Definitions \ref{d:jqm} and \ref{d:spjqm} are generalisations of sets defined for the CHSH scenario
in \cite{Craig:2006ny,Barnett:2007}.

Recall that positive semidefiniteness of a Hermitian matrix is equivalent to the matrix being the inner product matrix, the \textit{Gram matrix}, of a set of vectors spanning a Hilbert space $\cH$.  Using this, we can present an equivalent definition of the \textit{SPJQM} set which makes manifest its connection to the NPA heirarchy of conditions:

\begin{lemma}
\label{l:spjqm}
A behaviour  $(\Xi, \tcM, P)$ is in \textit{SPJQM} if and only if there exists a Hilbert space $\cH$ spanned by a set of vectors indexed by elements of the sigma algebra $2^\Xi$, $\{| X\ra\}_{ X \in 2^\Xi }$,  such that their inner products satisfy the following conditions:
\begin{itemize}
\item For any $X \in 2^\Xi$,
\begin{equation}
\label{e:spjqm_sum}
|X \ra = \sum_{\gamma \in X} |\{\gamma\} \ra\,;
\end{equation}
\item For each measurement $M\in \tcM$,
\begin{equation}
\label{e:decoherence}
\la X | Y \ra=  \, P(X\cap Y) \quad \forall \, X,Y \in \EA_M\,.
\end{equation}
In particular, if two outcomes, $X$ and $Y$, are disjoint  then $\la X | Y \ra= 0$: 
disjoint outcomes decohere. 
\end{itemize}
\end{lemma}

Although it is less simple to state, \textit{SPJQM} is physically more interesting than \textit{JQM}.  A theory would 
be problematic if composition of independent behaviours allowed by the 
theory yielded a behaviour not 
allowed by the theory (see \textit{e.g.}~\cite{Fritz:2013, Fritz:2012, Cabello:2012} for applications of this).  The \textit{JQM} condition is not composable: There exist behaviours for the CHSH scenario that are in \textit{JQM} but such that the composition of two independent copies is not \cite{Henson:2013}.  The \textit{SPJQM} condition on the other hand respects composition due to lemma \ref{l:compose}.  It is an open question whether the \textit{SPJQM} condition follows from the \textit{JQM} condition plus some simple principle expressing closure under composition of independent systems.
\section{Comparing the conditions}

\begin{lemma} Every behaviour for the CHSH (Bell (2,2,2)) scenario is in \textit{JQM}. 
\end{lemma}
For more general scenarios this is not the case, for example the composition of two independent PR boxes 
is not in \textit{JQM} \cite{Henson:2013}. 
\begin{corollary} \textit{Q} is a proper subset of \textit{JQM}.
\end{corollary}
\begin{lemma}
\label{l:tlm_bound}
If a behaviour for the CHSH scenario is in \textit{SPJQM}, then the experimental probabilities 
satisfy the Tsirelson-Landau-Masanes bound \cite{Tsirelson:1987,Landau:1988,Masanes:2003}.
\end{lemma}
The proofs of these two lemmas are in \cite{Craig:2006ny,Barnett:2007}.

\begin{lemma} $Q \subseteq SPJQM$
\end{lemma}
\begin{proof}
Suppose a behaviour is in \textit{Q}. Then projection operators $\{E^X\}$ indexed 
by the fine grained outcomes of basic measurements and
satisfying the conditions of definition \ref{d:q} exist. 
We define a  vector $|\{\gamma\}\ra$ for each atom $\{\gamma\}$ of $2^{\Xi}$:
\begin{equation}
|\{\gamma\}\ra := E^{X_1}\dots E^{X_p} |\psi \ra
\end{equation}
where $\{\gamma \} = X_1\cap\dots X_p$ and $X_i \in M_i$. We define vectors for 
each $X\in \tcC$ by addition:
\begin{equation}
|X \ra := \sum_{\gamma \in X}  |\{\gamma\} \ra\, .
\end{equation}
These vectors satisfy the conditions of lemma \ref{l:spjqm}.
\end{proof}

On comparison of the \textit{SPJQM} condition with the NPA conditions, similarities immediately become apparent.  Both require the existence of vectors with an inner product matrix that has a certain relationship to the experimental probabilities and a number of orthogonality relations.  Much of the work has been done by formulating the conditions in this way.  Based on this observation, the following relations can be discovered.

\begin{theorem}
\label{t:SPJQMinQ}
SPJQM $ \subseteq Q^1 $.
\end{theorem}
\begin{proof} 
Given a behaviour in \textit{SPJQM}, there exists 
a set of vectors obeying the conditions (\ref{e:spjqm_sum}-\ref{e:decoherence}).
To show that there exists a set of vectors with the properties (\ref{e:q1sum}) and (\ref{e:q1joint_measurement}) from definition \ref{d:q1} of $Q^1$  is an easy task:  the \textit{SPJQM} vectors range over all subsets of $\Xi$, while the $Q^1$ vectors range over a subset of these, so we can define each $Q^1$ vector as equal to the corresponding \textit{SPJQM} vector.  This immediately gives (\ref{e:q1sum}) from (\ref{e:spjqm_sum}) and (\ref{e:q1joint_measurement}) from (\ref{e:decoherence}). 
\end{proof}

\paragraph{} This theorem subsumes the earlier result given in lemma \ref{l:tlm_bound}: the Tsirelson-Landau-Masanes bound is implied by $Q^1$ \cite{Navascues:2008}.  It should be noted that, although
 the \textit{SPJQM} vectors are indexed by outcomes of all measurements in $\tcM$ and 
 not just the basic ones, one cannot use the same reasoning as above to prove that \textit{SPJQM} $ \subseteq \QAB $.  This is because the relations (\ref{e:tq1orthogonality})
 in the definition of $\QAB$  are \textit{not} implied by the conditions in the \textit{SPJQM} definition.  The significance of these extra orthogonality relations will be expanded on later.

\begin{theorem}
\label{t:tQ^1inSPJQM}
$\QAB \subseteq$ SPJQM \footnotemark.

\footnotetext{Miguel Navascues found this result for bipartite non-locality scenarios (private communication) using a new reformulation of $\QAB$ \cite{Acin:2013}.  The proof presented here for all joint measurement scenarios is based on his idea.}
\end{theorem}

\begin{proof} Consider a behaviour in $\QAB$. By definition \ref{d:tq1} there exists of a set of vectors $\{| X \ra \}_{X\in\tcC}$ associated to the outcomes of coarse-grained measurements, obeying the properties (\ref{e:tq1sum}-\ref{e:tq1orthogonality}).  From these vectors we will construct a set of vectors $\{ |X \ra \}_{X \in 2^\Xi}$ indexed by all subsets of the NC space, that obeys (\ref{e:spjqm_sum}) and contains the original set $\{| X \ra \}_{X\in\tcC}$ as its restriction to $\tcC$.  With this established, Lemma \ref{l:spjqm} shows that the behaviour is in \textit{SPJQM}.

We first construct  vectors $|\{\gamma \}\ra$ for all singleton sets of histories, $\gamma \in \Xi$. 
 For each basic measurement $M_i \in \cM$, we select, arbitrarily, one  outcome $X^*\in M_i$. 
 For every other outcome $X\in M_i$, $X\ne X^*$ 
 we define the projection operator,
\begin{equation}
\label{e:projectors}
E^X= \mbox{proj} [\cH^{X} ]\, ,
\end{equation}
where proj$[\cdot]$ projects onto the given subspace, and 
\begin{equation}
\label{e:subspaces}
\cH^X = \mbox{span} \{|X' \ra : X' \in \tcC \mbox{ s.t. } X' \subset X \}\, .
\end{equation}
That is, $\cH^X$ is the subspace spanned by the vectors for all measurement outcomes that imply the outcome $X$ for basic measurement $M_i$.
The orthogonality relation  (\ref{e:tq1orthogonality}) implies that $\cH^X$ and $\cH^Y$ 
are orthogonal subspaces for $X\ne Y$ and so  we have $E^X E^{Y} = \delta_{X,Y} E^X$
for all outcomes $X, Y \in M_i$, such that $ X, Y \ne X^*$.
  We then define $E^{X^*}$, for the remaining outcome $X^*$ of $M_i$, 
  by 
\begin{equation}
\label{e:proj_sum}
\sum_{X \in M_i} E^X = \mathbf{1},
\end{equation}
where $\mathbf{1}$ is the identity operator on the Hilbert space $\cH$.  Since the 
other $E^X$'s are 
orthogonal projectors, so is $E^{X^*}$. 
$E^{X^*}$  projects onto the orthogonal complement, $\cH^{X^*}$, of the space spanned by all the subspaces $H^X$ with $X \in M_i$ such that $X\neq X^*$. We have, therefore, 
\begin{equation}
 E^X E^{Y} = \delta_{X Y} E^X, \ \ \forall X\in M_i\,. 
\end{equation}
Thus, the set of operators  $\{E^X\}_{X \in M_i}$, indexed by the 
outcomes of $M_i$, forms a projective measurement. Note this does not prove that there is an ordinary quantum model, because these operators may not satisfy (\ref{e:q_commute}) .

Using these projectors, we define the vector for an element $\gamma \in \Xi$  as 
\begin{equation}\label{e:gammavector}
|\{\gamma \}\ra := E^{X_1} E^{X_2} ... E^{X_p} |\Xi \ra\, ,
\end{equation}
where $X_i$ is the outcome for $M_i$ such that $\{\gamma\} = X_1 \cap X_2\cap \dots 
X_p$.

We now show that the original $\QAB$ vectors are sums of the vectors $|\{\gamma\}\ra$.
Let $M\in \tcM$ be a joint measurement of the set of $m$ basic measurements $\{M_{i_1}, M_{i_2},\dots M_{i_m}\}$. Let $X \in M$ be a fine-grained outcome of $M$, $X = X_{i_1} \cap X_{i_2} \cap \dots X_{i_m}$
where $X_{i_k} \in M_{i_k}$. 
By using (\ref{e:proj_sum}) for all basic measurements \textit{not} in $M$, we obtain,
\begin{align}
\sum_{\gamma \in X} |\ \{ \gamma \} \ra  &= E^{X_{i_1}} E^{X_{i_2}} ... E^{X_{i_m}} |\Xi\ra \\
&=  E^{X_{i_1}} E^{X_{i_2}} ... E^{X_{i_m}} \sum_{X' \in M} |X'\ra.
\end{align}
where the second line is an application of (\ref{e:tq1sum}). Note that $X$ is one of the 
fine-grained outcomes of $M$ being summed over here.  Each $E^{X_{i_k}}$ either leaves 
$|X'\ra$ invariant or annihilates it depending on whether or not $X' \subset X_{i_k}$. Of all the outcomes of $M$
only $X$ itself is left invariant by all the projectors and so  
\begin{equation}
\sum_{\gamma \in X} |\gamma \ra  = | X \ra\, ,\ \ \forall X \in M\,.
\end{equation}
From (\ref{e:tq1sum}) we can see that this holds for all coarse-grained outcomes,
\begin{equation}
\sum_{\gamma \in X} |\gamma \ra  = | X \ra\, , \ \ \forall X\in \EA_M\, ,
\end{equation}
and, since we chose an arbitrary measurement $M$ to consider,  
this holds for all outcomes $X \in \tcC$. Therefore we have
 \begin{equation}
|X \ra = \sum_{\gamma \in X} | \{ \gamma \} \ra\,, \quad \forall X \in 2^\Xi \,.
\end{equation}
where this serves as the \textit{definition} of 
$|X\ra$ for each $X \in 2^\Xi \setminus \tcC $.

This set of vectors $\{ |X \ra \}_{X \in 2^\Xi}$ satisfies (\ref{e:spjqm_sum}) by construction.  Because it restricts to the original $\QAB$ vectors $\{ |X \ra \}_{X \in \tcC}$, condition (\ref{e:tq1joint_measurement}) in the definition of $\QAB$ ensures that it also obeys (\ref{e:decoherence}). Lemma \ref{l:spjqm} states that if there exists a set of vectors $\{ |X \ra \}_{X \in 2^\Xi}$ satisfying these two conditions, the behaviour is in \textit{SPJQM}.
\end{proof}

Note that the order in which the projectors are applied to the state $|\Xi\ra$
in (\ref{e:gammavector}) was chosen in the proof to be the reverse of the order of the 
labelling of the basic measurements, $1\dots p$, but any other order would have worked also.
Thus there is no unique set of vectors $\{|\{\gamma\}\ra \}$ satisfying the conditions.

\subsection{A complication: branching measurements}
\label{s:branching_measurements}

In the formalism above, we have not allowed measurement choices to depend on other measurement outcomes, though one can imagine setting up such a situation in a lab.  In the CHSH scenario, an example would be using the outcome 
of whatever measurement is done in the A wing to determine the measurement setting in the B wing, if the former is done in the causal past of the latter.  Such ``branching'' measurements play an important role in the alternative FLS formalism for contextuality \cite{Fritz:2013}, where they are necessary for the definition of Bell scenarios, and so a comparison with results in this formalism requires us to consider them. 

We will investigate the inclusion of these branching measurements and the altering of the non-contextuality condition to reflect the assumption that if measurements $M$ and $M'$ are jointly performable then they are jointly performable in any causal order -- either may causally precede the other, thus allowing branching, or they may be performed at spacelike separation -- without altering an underlying measure on the NC space.  This is to allow the classical communication between wings necessary for the branching choices. Note that this assumption is incompatible with an assumption that the basic measurements have 
fixed locations in spacetime and jointly performable basic measurements are spacelike
separated from each other. We argued above that such fixed spacetime locations for the basic measurements justifies our joint measurement scenario and behaviour framework. Including branching measurements means that the non-contextuality assumptions inherent in the 
joint measurement scenario and behaviour framework must be made without the argument from locality. 

We will show below that the altered non-contextuality condition
in the classical case is no stronger than 
JPM. However in the quantum case the altered non-contextuality condition 
makes a difference to our conclusions. 
To illustrate the reason for this, consider two outcomes in the CHSH scenario, $X=(a=+1,b=+1;x=0,y=0)$ and $Y=(a=-1,b=+1;x=0,y=1)$.  These are disjoint in the sense that $X \cap Y = \emptyset$
but there is no measurement partition in $\tcM$ that contains them both. Thus, 
a joint decoherence functional, $D_J$, for a behaviour for this scenario,
as defined in definition \ref{d:jqdf} or \ref{d:spjqm}, need not satisfy $D_J(X,Y)=0$: these outcomes need not decohere with respect to each other because they are not alternative outcomes of a possible 
measurement.  If, however, we add to 
the collection of possible measurements a branching measurement with both $X$ and $Y$ as alternative outcomes -- one in which the  outcome, $a$, of the measurement $x=0$ in the A wing
 is used to determine the $y$ setting in the B wing -- 
then the corresponding quantum non-contextuality condition requires that 
the joint decoherence functional restricts to a probability measure on the sigma algebra 
of outcomes for this branching 
 measurement. This implies $D_J(X,Y) =0$, an additional condition on the joint 
 quantum measure. All the extra conditions of this sort from branching 
 measurements will allow us to prove
 a stronger result. 

In general, given a joint measurement scenario $(\Xi, \tcM)$, the possible branching measurements 
form an extended set of measurement partitions. We will need to consider only branching 
measurements in which one decision is made; we call these single branching measurements. 
A single branching measurement consists of an initial (not necessarily basic) measurement $M^1\in \tcM$ 
and, depending on its outcome $X^1$, a second measurement $M^2(X^1)$ from the set of all those 
measurements in $\tcM$ that include $M^1$.  This defines a new measurement partition of $\Xi$, 
$M_b(M^1,M^2(X^1))$.  Explicitly, $M_b$ contains all $X^2 \in M^2(X^1)$ such that $X^2 \subset X^1$, for all $X^1 \in M^1$. We denote the set of all such single branching
measurements by $\tcMb$. Note that $\tcMb$ is uniquely determined given $\tcM$. 
Also $\tcM \subseteq \tcMb$ since  each measurement, $M\in \tcM$, is (trivially) a single
branching measurement where $M^1 = M$ and $M^2(X^1) = M$ for all $X^1$.
\footnotetext{One could also consider multiple branching measurements, the conceptually obvious 
(but notationally cumbersome) generalisation of single branching measurements.
However, this would have no effect of the argument below.  With such an altered definition of $\tcMb$ lemma \ref{l:branch} would still be true, and so the altered definition of $SPJDFb$ would still be equivalent to $\QAB$.  That is, once single-branching measurements have been added to set of possible measurements, adding multiple branching measurements does not alter the quantum non-contextuality condition.}

For each $M_b\in \tcMb$,  the sigma algebra $\EA_{M_b}$ is that generated by the fine-grained outcomes in the partition $M_b$. The set of all outcomes for all measurements in $\tcMb$ will be denoted $\tcCb :=  \bigcup_{M\in \tcMb} \,\EA_M$.

Only one consequence of these new definitions, expressed in the following lemma, will be salient to the discussion below.
\begin{lemma}
\label{l:branch}
If  $X, Y \in \tcC$ and there exists a basic measurement $M_i$ and outcomes $X', Y' \in M_i$ such that $X' \cap Y' = \emptyset$, $X \subset X'$ and $Y \subset Y'$ then there exists a single branching measurement $M^* \in \tcMb$ such that  $X, Y \in  \EA_{M^*}$.
\end{lemma}
\begin{proof}
If the premise holds, there exist two measurements $M_X, M_Y \in \tcM$ of which $X$ and $Y$ are, respectively, outcomes. Both $M_X$ and $M_Y$ must include basic measurement $M_i$. Consider the branching measurement $M^*$ in which $M^1 = M_i$ and 
$M^2(X') = M_X$, $M^2(Y') = M_Y$ and $M^2(Z) = M_i$ for
$Z\in M_i$,  $Z \ne X'$ and $Z\ne Y'$. Then $X, Y \in  \EA_{M^*}$.
\end{proof}

We have extended the joint measurement scenario $(\Xi,\tcM)$ to $(\Xi,\tcMb)$.
Given a behaviour  $(\Xi,\tcM,P)$ we also extend the probability function $P$ to cover the
larger set of outcomes $\tcCb$. Note here that each \textit{fine-grained} outcome for a branching measurement
is a fine-grained outcome for some joint measurement in $\tcM$, so the only elements in $\tcCb$ 
and not in $\tcC$ are coarse-grained outcomes of the new measurements.  This means that the probability function, $P$,  on $\tcC$ can be uniquely extended to a probability function on $\tcCb$, by applying the Kolmogorov sum rule.  We call the behaviour thus extended, $(\Xi, \tcMb, P)$,  the branching extension of 
$(\Xi, \tcM, P)$.

Does including these types of measurements  make any difference to the strength of the conditions discussed above?  That is, if in the definitions of our sets \textit{JPM}, \textit{Q}, \textit{JQM}, 
\textit{etc.}, we were to replace $\tcM$ with $\tcMb$ and $\tcC$ with $\tcCb$, would the sets remain the same?  For the set of classically non-contextual behaviours \textit{JPM} given by definition \ref{d:jpm}, it is not hard to see that nothing changes: the alternative definition would be the same except with equation (\ref{e:joint_measure}) replaced with
\begin{equation}
P(X) = P_J(X), \quad \forall X \in \tcCb\,, 
\end{equation}
whereas before $P_J$ is a probability measure over the sigma algebra $2^\Xi$.  Obviously this implies equation (\ref{e:joint_measure}) because $\tcC \subset\tcCb $. It is also implied by equation (\ref{e:joint_measure}), because as just discussed, for sets $X \in \tcCb$ which are not in $\tcC$, the value of $P(X)$ is determined by the probabilities of sets in $\tcC$ by the Kolmogorov rule, as is the value of $P_J(X)$.

It is a straightforward exercise to show that the same is true for the sets \textit{Q} and $Q^1$.  Here we spell out 
how this works for $\QAB$.  

\begin{lemma} 
\label{d:tq1_a}
A behaviour $(\Xi, \tcM, P)$, whose branching extension is
$(\Xi, \tcMb, P)$, is in $\QAB$ if and only if there exists a Hilbert space $\cH$ and a set of vectors $\{ |X \ra \}_{X \in \tcCb}$
indexed by all the coarse-grained outcomes of all branching measurements, that span $\cH$  and satisfy the following conditions:
\begin{itemize}
\item[(i)]  if $M_b\in \tcMb$ and $X,Y \in \EA_{M_b}$ are disjoint outcomes, $X\cap Y = \emptyset$, then
\begin{equation}
\label{e:tq1sum_a}
| X \cup Y \ra= |X\ra + |Y \ra\,;
\end{equation}
\item[(ii)]  for each branching measurement $M_b\in \tcMb$ and all outcomes $X, Y \in  \EA_{M_b}$ 
\begin{equation}
\label{e:tq1joint_measurement_a}
\la X | Y \ra= P(X \cap Y)\;.
\end{equation}
\end{itemize}
\end{lemma}
\begin{proof}
Suppose a behaviour is in $\QAB$. Construct vectors $\{ |X \ra\}_{X \in 2^\Xi} $ 
as in the proof of Theorem \ref{t:tQ^1inSPJQM}. These vectors 
satisfy (\ref{e:tq1sum_a}) when $X\cap Y = \emptyset$. 
If $X$ and $Y$ are fine-grained outcomes of $M_b$ then 
either they disagree on the outcome of the first measurement $M^1$ in the 
branching measurement in which case their corresponding vectors are 
orthogonal, or they agree on that first outcome in which case they are 
outcomes of the same measurement $M^2 \in \tcM$. Therefore
  (\ref{e:tq1joint_measurement_a}) holds for 
fine-grained outcomes of $M_b$. The vector for a coarse grained
outcome is the sum of the vectors for the fine-grained outcomes of
which it is the union and 
and so    (\ref{e:tq1joint_measurement_a})  also holds for coarse-grained outcomes of $M_b$.

Conversely, if vectors $\{ |X \ra \}_{X \in \tcCb}$ exist satisfying  (\ref{e:tq1sum_a}) and (\ref{e:tq1joint_measurement_a}) then those for $X\in \tcC$ satisfy conditions \textit{(i)} and 
\textit{(ii)} in definition \ref{d:tq1}.  It remains to show that these vectors 
satisfy  \textit{(iii)} in definition \ref{d:tq1}. Suppose $X, Y \in \tcC$ and there exists a basic 
measurement $M_i$ and outcomes $X', Y' \in M_i$ such that $X' \cap Y' = \emptyset$, 
$X\subset X'$ and $Y\subset Y'$. Then lemma \ref{l:branch} implies there is a branching
measurement $M^* \in \tcMb$ such that $X,Y \in \EA_{M^*}$. $X$ and $Y$ are disjoint and
so by (\ref{e:tq1joint_measurement_a}) we have
\begin{equation}\label{e:tq1orthogonality_a}
\la X | Y\ra = 0\,.
\end{equation}

\end{proof}

Thus, the conditions \textit{JPM}, \textit{Q}, $Q^1$ and $\QAB$, although originally expressed in terms of the 
joint measurements $\tcM$, also contain the information necessary to incorporate branching 
measurements in a non-contextual manner: once you have a \textit{JPM} -- or \textit{Q}, $Q^1$ or $\QAB$
-- model for the non-branching behaviour there is a unique way to extend it to a model of
the behaviour extended by adding all branching measurements. 

In contrast, however, the corresponding alteration of the set \textit{SPJQM} is not 
obviously the same.

\begin{definition}
\label{d:spjqmp} (SPJQM${}_b$)
A behaviour $(\Xi, \tcM, P)$, whose branching extension is
$(\Xi, \tcMb, P)$,
 is in the set \textit{SPJQM${}_b$} if there exists a joint decoherence functional, $D_J$ over $2^\Xi$ such that 
\begin{itemize}
\item[(i)] $D_J(X, Y) = P(X\cap Y) , \quad \forall \ X, Y \in \EA_{M_b},  \quad \forall \ M_b\in \tcMb$;
\item[(ii)] $D_J$ is strongly positive (SP).
\end{itemize}
\end{definition}
$SPJQM_b \subseteq SPJQM$, but as mentioned above there is reason to think that the two are not equivalent.  There is a modified version of lemma \ref{l:spjqm}.

\begin{lemma}
\label{l:spjqmp}
A behaviour  $(\Xi, \tcM, P)$, whose branching extension is
$(\Xi, \tcMb, P)$, is in \textit{SPJQM${}_b$} if and only if there exists a Hilbert space $\cH$ spanned by a set of vectors indexed by elements of the sigma algebra $2^\Xi$, $\{| X\ra\}_{X \in 2^\Xi}$,  such that their inner products satisfy the following conditions:
\begin{itemize}
\item[(i)] For any $X \in 2^\Xi$,
\begin{equation}
\label{e:spjqm_sum_a}
|X \ra = \sum_{\gamma \in X} |\{\gamma\} \ra\,;
\end{equation}
\item[(ii)] For each measurement $M_b\in \tcMb$,
\begin{equation}
\label{e:decoherence_a}
\la X | Y \ra=  \, P(X\cap Y) \quad \forall \, X,Y \in \EA_{M_b}\,.
\end{equation}
 In particular, if two outcomes $X$ and $Y$ of a measurement are disjoint then $\la X \, |\, Y \ra=0$: 
 disjoint outcomes decohere. 
\end{itemize}
\end{lemma}

\begin{theorem}
\label{t:SPJQMintQ}
\textit{SPJQM${}_b$}$ = \QAB $.
\end{theorem}
\begin{proof}
Lemmas \ref{l:spjqmp} and  \ref{d:tq1_a} imply that \textit{SPJQM${}_b$}$\subseteq \QAB$
and $\QAB\subseteq SPJQM_b$ can be proved following the proof of theorem \ref{t:tQ^1inSPJQM}.
\end{proof}

\section{Discussion}

We have found close relations between the NPA hierarchy of sets of experimental behaviours, and sets of
behaviours of interest in quantum measure theory.  In particular a quantum non-contextuality 
condition, \textit{SPJQM}, was defined, subsuming the previous definition for the CHSH scenario \cite{Craig:2006ny}.  This condition was found to be intermediate in strength between $Q^1$ and $\QAB$. 
Thus, if $\QAB$ is strictly larger than \textit{Q}, as indicated by the computational evidence, then 
 non-local correlations beyond those achievable in ordinary quantum mechanics are achievable within strongly positive quantum measure theory. However, it is possible that additional 
 constraints may arise within this framework as it is developed. 

We also saw that an interesting modification of \textit{SPJQM}, which we called \textit{SPJQM${}_b$},
is equal to the NPA set $\QAB$. The necessary change is to require 
that for any joint measurement behaviour allowed by the physics, 
 its branching extension is also physical. Including branching measurements, where the outcome
 of one measurement determines which future measurement is performed, imposes 
 more conditions because the quantum measure must restrict to a probability measure, and therefore
 decohere, on
 a larger collection of measurement subalgebras. The set
 \textit{SPJQM${}_b$} is therefore ``closed under wiring'' meaning that any behaviour created by applying ``classical operations'' -- coarse-graining measurement choices and/or outputs, post-selection, composition or ``branching'' as described above -- to a collection of behaviours in the set
 is also in the set  \cite{Acin:2013}. 
   It is not yet
 known whether   \textit{SPJQM${}_b$} is strictly contained in \textit{SPJQM} or whether
 they are equal. We will assume the former is the case, for the sake of the following 
 discussion.

Returning to our main motivation, the use of quantum analogues of the Fine trio 
 of conditions on a behaviour to discover quantum Bell causality, 
we see now that we have a choice for the ``joint measure" 
condition in the trio: \textit{SPJQM}or  \textit{SPJQM${}_b$}.  In contrast, 
as we saw, in the original Fine trio there is no choice to be made 
because the non-contextuality conditions are equivalent 
when the joint measure is a probability measure: $JPM \equiv JPM_b$. 
For the purpose of identifying a quantum Bell causality condition,  
the relevant scenarios are Bell scenarios where the 
jointly performable measurements are spacelike separated
which therefore rules out the possibility of branching measurements and it 
would seem therefore that \textit{SPJQM} is the appropriate condition.\footnotemark
It would be prudent to bear \textit{SPJQM${}_b$} in mind, however. We 
hope that the quantum non-contextuality condition will be a guide to discovering quantum 
Bell causality because the latter should imply the former. It is possible that 
if we find a scientifically fruitful
condition of quantum Bell causality -- perhaps by another route --  it will turn out to
 imply the stronger condition and then 
we will have discovered that \textit{SPJQM${}_b$} was the appropriate 
quantum non-contextuality condition after all. 

\footnotetext{This is the main reason we used a formalism based on the ideas of LSW \cite{Liang:2011} rather than FLS \cite{Fritz:2013}.  In the latter's framework, the natural definition of the existence of a strongly positive joint quantum measure would have automatically given \textit{SPJQM${}_b$} for Bell scenarios, because branching measurements are used by FLS in their ``Randall-Foulis product''.  
There would have been no way to define \textit{SPJQM} in the FLS formalism.}

Many other questions remain. The exact strength of the \textit{SPJQM} condition has not been determined: it is not known, for example, whether or not
it is equivalent to $Q^1$.  An exploration of  the relation of the weaker, but simpler \textit{JQM} condition to other principles such as local orthogonality and a proof that if \textit{JQM} is strengthened to $JQM_b$ 
by including branching measurements it implies local orthogonality will 
appear elsewhere \cite{Henson:2013}. The  assumption of strong positivity of the 
quantum measure was motivated by arguing that it guarantees composability 
of quantum measures on two independent systems.  If it can be proved that
strong positivity is necessary for composability then the assumption of strong positivity would have 
very strong physical motivation.  The connection to local orthogonality/consistent exclusivity might 
lead to new insights here. 
What of the higher levels of the NPA hierarchy?  As mentioned above, these conditions require the existence of a matrix with the same properties as the inner product matrix of vectors corresponding to strings of projection operators acting on the initial state in an ordinary quantum model. Strings of projection 
operators are known as  \textit{class operators} in the literature on decoherence functionals 
and are basic entities  from which decoherence functionals are constructed for sequences of
events ordered in time. For Bell scenarios therefore, it is possible that there is a connection 
between the higher level NPA conditions and decoherence functional based conditions arising from consideration of sequences of possible measurements.

\section*{Acknowledgements}
The authors thank Miguel Navascues for his comments on these ideas, especially those which led to theorem \ref{t:tQ^1inSPJQM}. Thanks are also due to Rafael Sorkin for detailed discussions of the work, and to Toni Acin, Adan Cabello, Belen Sainz and Tobias Fritz for helpful discussion on contextuality scenarios.  Thanks also to Alex Wilce and Rob Spekkens for comments on earlier versions of the work. The authors are grateful to Perimeter Institute for Theoretical Physics, Waterloo Canada and to the Institute for Quantum Computing, University of Waterloo, Ontario, Canada for hospitality at several stages of 
this research. Research at Perimeter Institute is supported in part by the Government of Canada through NSERC and by the Province of Ontario through MRI.
FD and PW are supported in part by COST Action MP1006. PW is supported by EPSRC grant No. EP/K022717/1.  JH is supported by a grant from the John Templeton Foundation. 

\begin{appendix}
\section{Definitions of the NPA heirarchy}

\label{a:npa_defs}

Here we demonstrate the equivalence  between the original definition of the sets of behaviours $Q^1$ and $Q^{1+AB}$ given by NPA and the definitions in the main text. 

Firstly we restate the original NPA form of the $Q^1$ condition
adapting it to our joint measurement scenarios. 
For  Bell scenarios, the first three conditions below are equivalent to equation (19) in \cite{Navascues:2008} and the remaining two conditions to (20) in \cite{Navascues:2008} for $n=1$.  NPA's ``null string'' $1$ is more naturally referred to as $\Xi$ in our terminology.

\begin{definition}
\label{d:q1npa}
A behaviour $(\Xi, \tcM, P)$  is in $Q^1_{\text{NPA}}$ if there exists a positive semi-definite matrix $\Gamma^1$ indexed by the set of all fine-grained outcomes $X \in M_i$ of all basic measurements $M_i \in \cM$, and by  $\Xi$, such that:
\begin{itemize}
\item
\begin{equation}
\label{e:q1npa_norm}
\Gamma^1_{\Xi,\Xi} = 1 ;
\end{equation}
\item If $M_i$ is a basic measurement and $X \in M_i$ is a fine-grained outcome,
\begin{equation}
\label{e:q1npa_a}
\Gamma^1_{\Xi,X} = P(X);
\end{equation}
\item  If basic measurements $M_i,M_j \in \cM$  are jointly performable and $M_i \neq M_j$, and $X \in M_i$ and $Y \in M_j$  are fine-grained outcomes, then
\begin{equation}
\label{e:q1npa_b}
\quad \Gamma^1_{X,Y} = P(X \cap Y);
\end{equation}
\item If $M_i$ is a basic measurement and $X \in M_i$ is a fine-grained outcome,
\begin{equation}
\label{e:q1npa_c}
\Gamma^1_{\Xi,X}= \Gamma^1_{X,X};
\end{equation}
\item If $M_i$ is a basic measurement and $X,Y \in M_i$ are fine-grained outcomes,
\begin{equation}
\label{e:q1npa_d}
\Gamma^1_{X,Y} = 0 \mbox{ if }X\neq Y.
\end{equation}
\end{itemize}
\end{definition}

The translation to the terminology of NPA is fairly direct.  For example, for two fine-grained outcomes $X$ and $Y$ taken from jointly performable basic measurements,  $P(X \cap Y)$ corresponds to $P(a,b)$, because the only jointly performable pairs of measurements in bipartite Bell scenarios are those taken from different wings.  The following reformulation brings us closer to that used in the main text:

\begin{lemma}
\label{l:q1npa}
A behaviour $(\Xi, \tcM, P)$  is in $Q^1_{\text{NPA}}$ if and only if there exists a positive semi-definite matrix $\Gamma^1$ indexed by the set of all fine-grained outcomes $X \in M_i$ of all basic measurements $M_i \in \cM$, and by $\Xi$, such that:
\begin{itemize}
\item For all basic measurements $M_i \in \cM$ and for all $X \in M_i \cup \{ \Xi \}$,
\begin{equation}
\label{e:q1npa2_sum}
\Gamma^1_{X,\Xi} = \sum_{Y \in M_i} \Gamma^1_{X,Y} ;
\end{equation}
\item  If basic measurements $M_i,M_j \in \cM$  are jointly performable and $X \in M_i$ and $Y \in M_j$  are fine-grained outcomes, then
\begin{equation}
\label{e:q1npa2_joint}
\quad \Gamma^1_{X,Y} = P(X \cap Y);
\end{equation}
\end{itemize}
\end{lemma}
\begin{proof}
First we show that (\ref{e:q1npa_norm}-\ref{e:q1npa_d}) imply (\ref{e:q1npa2_sum},\ref{e:q1npa2_joint}).  Equation (\ref{e:q1npa2_joint}) is the same as (\ref{e:q1npa_b}) in the case that $M_i \neq M_j$.  When $M_i=M_j$ and $X \neq Y$  (\ref{e:q1npa2_joint}) follows from (\ref{e:q1npa_d}) (in this case $X \cap Y=\emptyset$), and for $X=Y$  (\ref{e:q1npa2_joint}) follows from (\ref{e:q1npa_c}) and (\ref{e:q1npa_a}).  To obtain (\ref{e:q1npa2_joint}) when $X=\Xi$, we apply (\ref{e:q1npa_norm}) and then (\ref{e:q1npa_a}):
\begin{equation}
\Gamma^1_{\Xi,\Xi} = 1 = \sum_{Y \in M} P(Y) = \sum_{Y \in M} \Gamma^1_{\Xi,Y}
\end{equation}
for all basic measurements $M\in \cM$. Otherwise $X$ is some outcome of the basic measurement $M_i$, and we apply (\ref{e:q1npa_a}) and then (\ref{e:q1npa2_joint}):
\begin{equation}
\Gamma^1_{X,\Xi} = P(X) = \sum_{Y \in M_i} P(X \cap Y) = \sum_{Y \in M_i} \Gamma^1_{X,Y}
\end{equation}
Now we show the converse, that (\ref{e:q1npa2_sum},\ref{e:q1npa2_joint}) imply (\ref{e:q1npa_norm}-\ref{e:q1npa_d}).  Equation (\ref{e:q1npa2_joint}) implies (\ref{e:q1npa_b}) and (\ref{e:q1npa_d}).  Then (\ref{e:q1npa_d}) and (\ref{e:q1npa2_sum}) together imply (\ref{e:q1npa_c}).  Then (\ref{e:q1npa_c}) and (\ref{e:q1npa2_joint}) together imply (\ref{e:q1npa_a}). Finally (\ref{e:q1npa_a}) and (\ref{e:q1npa2_sum}) together imply (\ref{e:q1npa_norm}).
\end{proof}

It is an elementary result that a matrix is positive semidefinite (has no negative eigenvalues) if and only if it is the inner product matrix (or ``Gram matrix'') of some set of vectors in a Hilbert space, and we rely on this to prove the following result. 

\begin{lemma}
 $Q^1=Q^1_{\text{NPA}}$.
\end{lemma}
\begin{proof}
Consider a behaviour satisfying the conditions of definition \ref{d:q1} 
and define  $\Gamma^1_{X,Y} := \la X | Y \ra$ for $X, Y$ fine-grained outcomes of basic 
measurements or $\Xi$. Then
 $\Gamma^1$ is positive semidefinite,  and equations (\ref{e:q1npa2_sum},\ref{e:q1npa2_joint}) are implied by (\ref{e:q1sum},\ref{e:q1joint_measurement}) respectively.

Conversely, consider a behaviour in $Q^1_{\text{NPA}}$.  There exists a set of vectors indexed by the set of fine-grained basic measurement outcomes and $\Xi$, spanning a Hilbert space $\cH$, such that $\Gamma^1_{X,Y} = \la X | Y \ra$.  We can extend this set of vectors to one indexed by $\cC$ -- that is, all coarse-grained outcomes of basic measurements -- by summation 
according to (\ref{e:q1sum}). This is consistent 
with the property (\ref{e:q1npa2_sum}) of the original vectors.  Condition (\ref{e:q1joint_measurement}) follows from (\ref{e:q1npa2_joint}), noting that if this equation holds for the fine-grained basic measurement outcomes then it will also hold for unions of them, the coarse-grained outcomes.
\end{proof}

Very similar reasoning can be applied to the case of $\QAB$.

\end{appendix}
\bibliographystyle{../jhep}
\bibliography{refs0.2}

\end{document}